\documentclass[journal]{IEEEtran}

\usepackage{latexsym}
\usepackage{graphicx}
\usepackage{array}
\usepackage{amsmath}
\usepackage{amsfonts}
\usepackage{amssymb}
\usepackage{amsthm}

\newtheorem{thm}{Theorem}
\newtheorem{lemma}{Lemma}
\newtheorem{prop}{Proposition}

\newcommand{\by} {\boldsymbol{y}}

\newcommand{\bh} {\boldsymbol{h}}

\newcommand{\bw} {\boldsymbol{w}}

\newcommand{\bxi} {\boldsymbol{\xi}}

\def\bal#1\eal{\begin{align}#1\end{align}}
\newcommand{\bp} {\begin{proof}}
\newcommand{\ep} {\end{proof}}

\newcommand{{\bRF}} {\right\}}

\begin{document}

\title{\huge Structural Design of Non-Uniform Linear Arrays for Favorable Propagation in Massive MIMO}

\author{Elham Anarakifirooz, Sergey Loyka

\vspace*{-1.5\baselineskip}

\thanks{E. Anarakifirooz and S. Loyka are with the School of Electrical Engineering and Computer Science, University of Ottawa, Canada, e-mail: sergey.loyka@ieee.org}
}

\maketitle


\vspace*{-1.2\baselineskip}
\begin{abstract}
Favorable propagation (FP) for massive multiple-input multiple-output (MIMO) systems with uniform and non-uniform linear arrays is considered. A gap in the existing FP studies of uniform linear arrays is identified, which is related to the existence of grating lobes in the array pattern and which results in the FP condition being violated, even under distinct angles of arrival. A novel analysis and design of \textit{non-uniform} linear arrays are proposed to cancel grating lobes and to restore favorable propagation for \textit{all} distinct angles of arrival. This design is consistent with the popular hybrid beamforming paradigm and extends to multipath channels and arrays of directional elements.

\end{abstract}

\vspace*{-0.4\baselineskip}

\vspace*{-0.6\baselineskip}
\section{Introduction}

Massive MIMO is widely accepted as one of the key technologies for 5G and beyond. It provides significant improvements in spectral and energy efficiencies as well as simplified processing in multi-user environments, due to a phenomenon known as "favorable propagation", whereby the channel vectors of different users become orthogonal to each other as the number of antennas increases \cite{Marzetta-16}-\cite{Masouros-15}. While they are not exactly orthogonal to each other in practice, it has been shown theoretically and experimentally that the FP holds approximately in many scenarios  of practical interest so that the benefits of massive MIMO can be exploited \cite{Marzetta-16}-\cite{Martinez-18}.

Favorable propagation for uniform linear arrays (ULA) was studied analytically in \cite{Marzetta-16}-\cite{Masouros-15} and experimentally in \cite{Hoydis-12}-\cite{Martinez-18}. Antenna array geometry and propagation environment along with users' locations were identified as the key factors determining the existence or non-existence of the FP. It was concluded that, for a fixed antenna element spacing and line-of-sight (LOS) propagation, the FP holds asymptotically (as the number of antenna elements increases without bound) as long as users have distinct angles of arrival (AoA). In this Letter, we show that this conclusion is based on the implicit assumption (not mentioned in the above studies) that there are no grating lobes (GL) in the array pattern and that it fails to hold if GLs are present and some users align with their directions.

Larger element spacing under fixed number of elements (i.e. fixed complexity/cost) is desirable to increase the array spatial resolution (or, equivalently, to decrease its beamwidth), i.e. its ability to resolve nearby users and hence to cancel inter-user interference (IUI). However, this has a major drawback as grating lobes appear in the array pattern for larger spacing \cite{VanTrees-02}\cite{Hansen-98} so that different users appearing at the main beam and GLs directions cannot be resolved, which creates significant IUI at those directions. This has a profound negative impact on favorable propagation, even under distinct AoAs. The larger the antenna spacing, the more grating lobes appear so that more user directions should be banned to maintain the FP.

An experimental observation of the grating lobe's impact on IUI and favorable propagation was  reported in \cite{Martinez-18}, but, to the best of our knowledge, no comprehensive \textit{analysis} is available in the literature, so that the existing FP analysis for ULAs is incomplete in this respect. No design to eliminate the impact of GLs on favorable propagation was proposed either. A number of designs of traditional antenna arrays (not massive MIMO) have been proposed to suppress grating lobes \cite{Krivosheev-15}. However, these designs do not target FP explicitly and it is far from clear whether the FP condition is satisfied for these designs (note that the absence of GLs does \textit{not} guarantee the FP).

To address these issues, we present a rigorous analysis of grating lobes' impact on favorable propagation and  propose a novel non-uniform linear array (NULA) design that effectively cancels grating lobes and also ensures that the FP holds for any element spacing and any distinct AoAs. The proposed NULA design is block-partitioned, whereby each block (subarray) is an ULA but the overall array is not uniform. In order to cancel GLs, we show that the number of subarrays (blocks) and their spacing have to be carefully selected so that, asymptotically, GLs are cancelled by nulls in the block array factor and thus the FP is restored for any distinct AoAs. A rigorous analysis is presented to find proper number of subarrays and their spacing; it includes some tools from number theory, which, to the best of our knowledge, have not been used before in the antenna array or massive MIMO literature.  While the actual number of grating lobes and their directions do depend on the main beam direction (i.e. beam steering) \cite{VanTrees-02}\cite{Hansen-98}, the subarray-based design proposed here is independent of it, so it can accommodate beam steering as well since grating lobes are canceled for any direction of the main beam. This design is also consistent with the popular hybrid beamforming paradigm \cite{Puglielli-16}. It allows one to increase the spatial resolution (and hence to accommodate more users) under a fixed array complexity/cost by increasing element spacing and, at the same time, to avoid an increase in IUI and to restore the FP by cancelling GLs (which appear due to a larger element spacing).

The main contributions of this Letter are in Propositions \ref{prop.alpha} and \ref{prop.alpha.asymp} (the IUI leakage factor for block-partitioned NULA) and Theorem \ref{thm.NULA} (the NULA design to eliminate grating lobes and restore favorable propagation, for any element spacing, beam steering and distinct AoAs). An extension of these results to multipath/blocked-LOS channels and to arrays of directional elements is discussed in Sec. V.

\vspace*{-.4\baselineskip}
\section{Channel Model and Favorable Propagation}
\label{sec.Channel Model}

Let us consider a frequency-flat Gaussian MIMO channel, where $M$ independent single-antenna users transmit simultaneously to a base station (BS) equiped with an $N$-element antenna array:
\bal
\label{eq.ch.model}
\by = \bh_1 x_1 + \sum_{i=2}^{M} \bh_i x_i + \bxi
\eal
where $\bh_i, x_i$ are the channel vector and the transmitted signal of user $i$, $i=1,...,M$; $\by, \bxi$ are the received BS signal and noise vectors, respectively; $|\bh|, \bh'$ and $\bh^+$ denote Euclidean norm (length), transposition and Hermitian conjugation, respectively, of vector $\bh$. The white noise is Gaussian circularly symmetric, of zero mean and variance $\sigma_0^2$ per Rx antenna. Frequency-selective channels can be considered via an OFDM-type approach, to which our results can be extended as well.

To simplify the decoding process, the BS uses linear processing with matched filter beamforming\footnote{Under the FP as in \eqref{FP_cond}, the MF beamformer performance is the same as that of zero forcing (ZF), minimum mean square error (MMSE) and successive interference cancellation (SIC) beamformers since users become "orthogonal" to each other \cite{Marzetta-16}. The significant advantage of the MF beamformer is its smaller computational complexity and higher robustness (decoding different users are independent of each other, which also allows for parallel/distributed implementation).} $\bw = \bh_1/|\bh_1|$ to decode user 1, treating other users' signals as interference. Hence, its SINR can be expressed as follows:
\bal
\label{eq2}
\mbox{SINR} &=\frac{|\bh_1|^2\sigma_{x_1}^2}{|\bh_1|^{-2} \sum_{i=2}^{M}|\bh_1^+\bh_i|^2\sigma_{x_i}^2 + \sigma_0^2}\\
&= \gamma_1\left(\sum_{i=2}^{M}|\alpha_{iN}|^2\gamma_i + 1\right)^{-1},\quad \alpha_{iN}=\frac{\bh_1^+\bh_i}{N}
\eal
where $\sigma_{xi}^2$ and $\gamma_i=|\bh_i|^2 \sigma_{xi}^2/\sigma_0^2$ are the signal power and the SNR of user $i$; the channel is normalized so that $|\bh_i|^2=N$ (the propagation path loss is absorbed into the Rx SNR $\gamma_i$). Using this simplified decoding method, the SINR cannot exceed the single-user SNR $\gamma_1$,
\bal
\mbox{SINR} \le \gamma_1
\eal
and this maximum is attained when the users’ channels become orthogonal to each other,
\bal
\mbox{SINR} \to \gamma_1\  \operatorname{if}\ \alpha_{iN} \to 0\ \forall\ i>1
\eal
as the number $N$ of antenna elements increases and all SNRs stay uniformly bounded, i.e. $\gamma_i \le \gamma_{max} < \infty$ for some $\gamma_{max}$ independent of $N$ (where $\gamma_i$ may depend on $N$). This is known as (asymptotically) favorable propagation condition.
When the number of users is finite, the FP condition can be expressed in two equivalent ways:
\bal
\label{FP_cond}
\lim_{N\to\infty}\alpha_N^2=0\quad\Leftrightarrow\quad \lim_{N\to\infty}|\alpha_{iN}|=0\ \forall i >1
\eal
where $\alpha_N^2 = \sum_{i=2}^{M}|\alpha_{iN}|^2$ characterizes the total interference leakage and $|\alpha_{iN}|^2$ represents IUI power ”leakage” from user $i$ to the main user. If all users have the same SNR ($\gamma_i=\gamma_1$), then the SINR simplifies to:
\bal
\mbox{SINR} = (\alpha_N^2 + \gamma_1^{-1})^{-1} \le \gamma_1
\eal
and, under favorable propagation, the upper bound is attained, $\mbox{SINR} = \gamma_1$. While in practice the number $N$ of elements is always finite and $\alpha_N^2$ is never exactly zero, the FP is closely approached if $\alpha_N^2 \ll \gamma_1^{-1}$ so that $\mbox{SINR} \approx \gamma_1$. This justifies the asymptotic analysis $N \to \infty$ from the practical perspective, since, if the asymptotic FP \eqref{FP_cond} holds, it follows from the limit definition that there exists a sufficiently large $N$ for which $\alpha_N^2 \ll \gamma_1^{-1}$ and thus $\mbox{SINR} \approx \gamma_1$. This is no longer the case if the FP does not hold.

\section{Favorable propagation for ULAs}
\label{sec.ULA}

Favorable propagation for uniform linear arrays has been investigated analytically and experimentally \cite{Marzetta-16}-\cite{Martinez-18}. It was concluded that FP holds asymptotically for ULAs of fixed element spacing (Case 1) under LOS propagation conditions and distinct users' AoAs \cite{Marzetta-16}-\cite{Ngo-14} but does not hold if the antenna array size is fixed (Case 2) \cite{Masouros-15}, so that the element spacing decreases when $N$ increases. However, the Case 1 conclusion is based on the implicit assumption, not stated in the above studies, that there are no grating lobes in the array factor and, as we show below, it does not hold if GLs are present\footnote{as a side remark, note that there are no GLs in Case 2 considered in \cite{Masouros-15} (since element spacing there becomes sufficiently small as $N$ increases under a fixed array size) yet the FP does not hold either, i.e.,  the absence of GLs does \textit{not} guarantee the FP but their presence does break down the FP.}. In this paper, we consider only Case 1.

To see why FP does not hold if GLs are present, observe that \cite[eq. (7.17)]{Marzetta-16} holds only if its denominator is not zero  but this condition is violated if GLs are present in the antenna array pattern, even under distinct users' AoAs, as \eqref{eq5}, \eqref{eq6} and the examples below show. This observation also applies to \cite[eq. (18)]{Ngo-14} and \cite[eq. (34), (42)]{Chen-13}, so that Propositions 4 and 5 in \cite{Chen-13} hold if there are no GLs and may not hold otherwise.

\begin{figure}
\begin{center}
\includegraphics[width=1.5in]{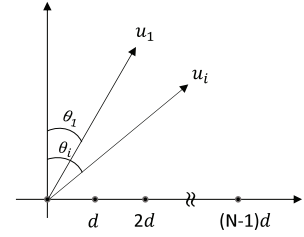}
\vspace*{-0.5\baselineskip}
\caption {An illustration of ULA$(N,d)$ geometry, where user $i$ AoA is $\theta_i$ while user 1 (main user) is at $\theta_1$, all measured from the broadside; $-\pi/2 \le \theta_1,\ \theta_i \le \pi/2$. When user 1 is decoded, user $i$ is a source of interference.}
\label{fig1}
\end{center}
\end{figure}

\begin{figure}
\vspace*{-0.5\baselineskip}
\includegraphics[width=3.5in]{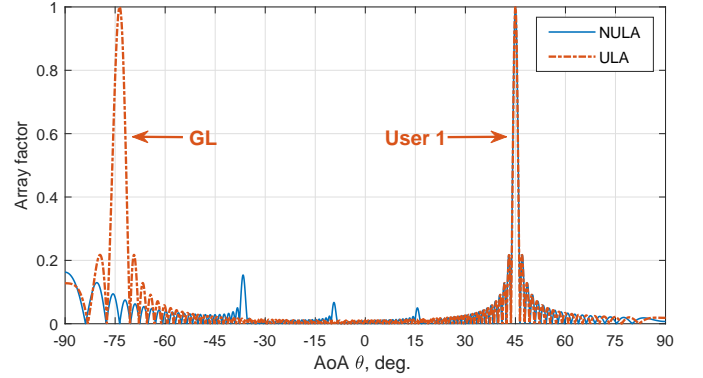}
\vspace*{-1.5\baselineskip}
\caption {Array factor of a ULA with $N=100$, $d=0.6$. While the main beam is at $\theta_1=45^o$ (where the ULA and NULA factors overlap), note the presence of a GL at $\theta_2 \approx -74^o$. The proposed NULA design with $N_b=25,\ N=4,\ p=21$ partially cancels this GL (see Theorem \ref{thm.NULA}) while preserving the main beam.}
\label{grating_lobe}
\vspace*{-1\baselineskip}
\end{figure}

To analyze the impact of grating lobes on favorable propagation, we follow \cite{Marzetta-16}-\cite{Masouros-15} and consider first LOS-dominated environment, as in e.g. mmWave/THz systems where LOS is essential to maintain a proper SNR \cite{Larson-18}, and where users have distinct AoAs (otherwise, the FP does not hold). In the far-field, the normalized channel vector of $i$-th user for $N$-element ULA of omnidirectional elements (e.g. vertical dipoles for a horizontal array as in Fig. \ref{fig1}) is:
\bal
\label{eq3}
&\bh_i =[e^{j\psi_{0i}},\cdots,e^{j\psi_{(N-1)i}}]',\quad i=1,\cdots,M\\ \notag
&\psi_{ni} =2\pi nd\sin(\theta_i),\quad n=0,\cdots,N-1
\eal
where $d$ is the element spacing (measured in wavelengths), $\theta_i$ is the AoA of user $i$ signal of which there are $M$, $-\pi/2 \le \theta_i \le \pi/2$; $\psi_{ni}$ is the phase shift at element $n$ with respect to 1st element, see Fig. \ref{fig1}. Omnidirectional elements (e.g. vertical dipoles) are assumed here for simplicity and the results are straightforward to extend to directional elements as well, see Sec. V.  All users are assumed to be in the front half-plane, $-\pi/2 \le \theta_i \le \pi/2$, and there is no backward radiation (see footnote 3). For further reference, we use $\operatorname{ULA}(N,d)$ to denote a uniform linear array with $N$ omnidirectional elements and element spacing $d$. Using the system model in \eqref{eq.ch.model}, the inter-user interference leakage from user $i$ to the main user ($i=1$) can be expressed as follows:
\bal
\label{eq5}
\alpha_{iN}&= \frac{1}{N} \sum_{n=0}^{N-1} e^{j n \Delta\psi_i} =\frac{\sin(N\Delta\psi_i/2)}{N\sin(\Delta \psi_i/2)} e^{j(N-1)\Delta\psi_i/2}
\eal
where $\Delta\psi_i = 2\pi d (\sin\theta_i - \sin\theta_1)$. Notice that $\lim_{N\to\infty} \alpha_{iN} =0$, i.e. the FP holds, provided that $\sin(\Delta \psi_i/2)\neq 0$. The latter condition may be violated even if $\theta_i \neq \theta_1$ (distinct AoAs), e.g. if $d=1,\ \theta_1=0,\ \theta_i = \pm 90^o\neq \theta_1$ so that $\Delta \psi_i= \pm 2\pi$, $\sin(\Delta \psi_i/2)= 0$ and hence $|\alpha_{iN}|=1$ for any $N$. This represents a grating lobe in the array pattern, see e.g. \cite{VanTrees-02}\cite{Hansen-98}. In general, GL directions $\phi_k$ correspond to zero denominator in \eqref{eq5}, i.e. $\sin(\Delta\psi_i/2)=0$, and, for a given $\theta_1$, can be found from $\Delta\psi_i=2\pi k$ with $\theta_i=\phi_k$, which is equivalent to
\bal
\label{eq.phi_k}
\sin(\phi_k)=\sin(\theta_1)+ k/d
\eal
where $k$ is the GL index, $k= \pm 1, \pm 2, ...$. Since $|\sin \phi_i| \le 1$, there exist no grating lobes if
\bal
\label{eq.gl.cond}
d(1+|\sin \theta_1|) < 1
\eal
In this case, the results in \cite{Marzetta-16}-\cite{Ngo-14} do hold for any distinct $\theta_i$, but they may fail to hold if \eqref{eq.gl.cond} is not satisfied. Indeed, using \eqref{eq5}, it follows that, under distinct AoAs $\theta_i \neq \theta_1$,
\bal
\label{eq6}
\lim_{N\to\infty}|\alpha_{iN}|=\left\{\begin{array}{ll}
1, \mbox{if}\ d(1+|\sin \theta_1|) \ge 1\ \&\ \theta_i=\phi_k\\
0,\ \mbox{otherwise}.
\end{array}\right.
\eal
Note a dichotomy here: the limit is either zero or one. The latter case gives the conditions when the FP fails to hold: $d(1+|\sin \theta_1|) \ge 1$ is the condition for the GL existence, and $\theta_i=\phi_k$ is the condition for $i$-th user AoA to coincide with $k$-th GL direction.

To determine the number of grating lobes, observe that $|\sin(\phi_k)|\le 1$ and use \eqref{eq.phi_k} to obtain the range of $k$:
\bal
\label{eq7}
k \ge k_{min} &= -\lfloor d(1+\sin(\theta_1))\rfloor\ge  -\lfloor 2d\rfloor\\
\label{eq8}
k \le k_{max} &= \lfloor d(1-\sin(\theta_1))\rfloor\le\lfloor 2d\rfloor
\eal
where $\lfloor \cdot \rfloor$ is the floor function, so that $I_k = \{\{k_{min},\cdots, k_{max}\} - \{0\}\}$ is the GL index set, $k=0$ represents the main beam and hence is excluded; $I_k$ is an empty set if there are no GLs. Thus, the number $K$ of GLs does not exceed $\lfloor 2d\rfloor$: $K = k_{max}-k_{min} \le \lfloor 2d\rfloor$, and there are no grating lobes if $d<1/2$, for any $\theta_1$.

\textbf{Examples:} To illustrate cases where the FP fails due to the presence of GLs, let $d=0.6$, $\theta_1=45^o$ and $\sin\theta_i=\sin\theta_1-1/d$ so that $\theta_i \approx -74^o \neq \theta_1$, i.e., distinct AoAs, but $\Delta\psi_i=-2\pi$ and the denominators in \eqref{eq5} as well as in \cite[eq. (18)]{Ngo-14}\cite[eq. (34), (42)]{Chen-13} are all zero and $|\alpha_{iN}|=1$ for any $N$,  i.e. the FP fails to hold even though the AoAs are distinct. This can be explained via the array factor shown in Fig. \ref{grating_lobe}, where the main beam is at $\theta_1=45^o$ to follow user 1 while the grating lobe appears at $-74^o$, so that, if another user is at the latter direction, it cannot be discriminated from the 1st user and hence the FP fails to hold. Note from \eqref{eq6} that the FP may fail to hold even if $d=1/2$ (as in \cite[eq. (7.17)]{Marzetta-16}), e.g. if $\theta_1=90^o$, $\theta_i=\phi_{-1}=-90^o \neq \theta_1$ so that $\Delta \psi_i = -2\pi$, $\sin(\Delta \psi_i/2) = 0$ and $|\alpha_{iN}|=1$ for any $N$ (but the FP always holds if $d<1/2$).
Even if the main user is at broadside, i.e. $\theta_1=0^o$, grating lobes appear if $d\ge 1$ and the FP fails to hold, even under distinct AoAs\footnote{since the ULA made of isotropic or omnidirectional elements is not able to discriminate users' signals coming from opposite directions, i.e. $\theta_1$ and $180^o-\theta_1$, the FP also fails to hold under distinct but opposite AoAs, for any element spacing. However, in practice this is not that important since backward radiation is usually eliminated due to an element pattern or an array design including a ground plane. Thus, it is not considered here; instead, we assume that $-\pi/2 \le \theta_1,\ \theta_i \le \pi/2$.}, e.g. if $d=1,\ \theta_1=0^o,\ \theta_i=\pm \pi/2$, which also results in zero denominators in \cite{Chen-13}\cite{Ngo-14}. In general, the larger $d$, the more grating lobes emerge \cite{VanTrees-02}\cite{Hansen-98} and favorable propagation fails to hold if users' AoAs, being distinct from each other, coincide with the GL directions. Eq. \eqref{eq6} below gives a precise condition for this to happen.

\section{Non-uniform Linear Array Design for FP}
\label{sec.NULA}

Motivated by the above analysis of grating lobes and their impact on the FP, we present below a structural design of nonuniform linear arrays that eliminates  all GLs and guarantees the FP to hold for any distinct AoAs and any element spacing $d$. To this end, let us consider a block-partitioned non-uniform array as in Fig. \ref{fig2}: the overall NULA consists of $N_b$ subarrays (blocks), which are $\operatorname{ULA}(N,d)$ and which are arranged in the $\operatorname{ULA}(N_b,D)$ block-wise pattern. The overall NULA pattern is a product of the subarray factor of $\operatorname{ULA}(N,d)$ and the block array factor of $\operatorname{ULA}(N_b,D)$ (where each block is replaced with an omnidirectional element), see e.g. \cite{VanTrees-02}\cite{Hansen-98}. Thus, the GLs in the subarray factor can be cancelled with nulls in the block array factor, as explained below. Finding proper $N_b$ and the subarray spacing $\Delta D$ are crucial  to cancel grating lobes and hence to achieve favorable propagation for any distinct AoAs and any $d$.

\begin{figure}
\begin{center}
\includegraphics[width=2.5in]{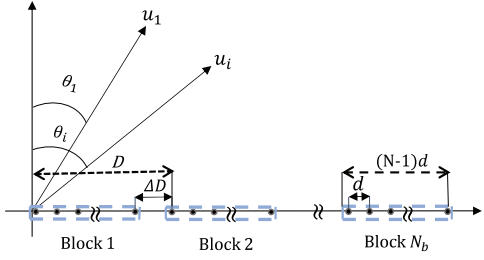}
\caption {Block-partitioned NULA of $N_b$ subarrays (blocks) $\operatorname{ULA}(N,d)$ with subarray spacing $\Delta D$; $D=(N-1)d+\Delta D$, and $(N-1)D$ is the subarray length. The AoAs are distinct, $\theta_1 \neq \theta_i$, and $-\pi/2 \le \theta_1,\ \theta_i \le \pi/2$.}
\vspace*{-1\baselineskip}
\label{fig2}
\end{center}
\end{figure}

Due to the block-wise symmetry of the structure, the overall channel vector $\bh_i$ of the NULA for user $i$ can be expressed as follow:
\bal
\label{eq10}
\bh_i = \bh_{si}\otimes \bh_{bi}
\eal
where $\otimes$ denotes Kronecker product; $\bh_{si}$ and $\bh_{bi}$ represents the channel vector of the subarray $\operatorname{ULA}(N,d)$ and of the block array $\operatorname{ULA}(N_b,D)$, respectively, where $\bh_{si}$ is as in \eqref{eq3} and $\bh_{bi}$ is
\bal
\label{eq.hbi.nula}
&\bh_{bi} =[e^{j\psi_{b,0i}},\cdots,e^{j\psi_{b,(N_b-1)i}}]',\quad i=1,\cdots,M \\\notag
&\psi_{b,ni} =2\pi n D \sin(\theta_i),\quad n=0,\cdots,N_b-1
\eal
For further use, the interference leakage terms of the NULA are defined as follows:
\bal
\label{eq.alpha.def}
\alpha_{iN} = \frac{\bh_1^+\bh_i }{N_bN},\ \alpha_{siN} = \frac{\bh_{s1}^+\bh_{si}}{N},\ \alpha_{bi} = \frac{\bh_{b1}^+\bh_{bi}}{N_b}
\eal
where $\alpha_{siN}$ and $\alpha_{bi}$ represent the respective terms for a single subarray $\operatorname{ULA}(N,d)$ and the block array $\operatorname{ULA}(N_b,D)$ while $\alpha_{iN}$ represents the overall leakage.

The following proposition is instrumental in establishing the FP for the block-partitioned NULA.
\begin{prop}
\label{prop.alpha}
Let $\bh_i$ have the Kronecker structure as in \eqref{eq10}. Then, $\alpha_{iN}$ can be expressed and bounded as follows:
\bal
\label{eq11}
&\alpha_{iN} =\alpha_{siN} \alpha_{bi}\\
\label{ee.alph.ineq}
&|\alpha_{iN}|\le \min\{|\alpha_{siN}|,| \alpha_{bi}|\}
\eal
\begin{proof}
Note the following:
\bal\notag
\label{eq.alpha_i}
\alpha_{iN} &= (N_b N)^{-1} \bh_1^+\bh_i\\ \notag
&= (N_b N)^{-1}(\bh_{s1} \otimes \bh_{b1})^+ (\bh_{si} \otimes \bh_{bi})\\ \notag
&=(N_b N)^{-1}(\bh_{s1}^+ \otimes \bh_{b1}^+) (\bh_{si} \otimes \bh_{bi})\\ 
&=(N^{-1} \bh_{s1}^+ \bh_{si})(N_b^{-1} \bh_{b1}^+ \bh_{bi})=\alpha_{siN} \alpha_{bi}
\eal
where 3rd and 4th equalities are due to the properties of Kronecker products \cite{Magnus-99}. The inequality in \eqref{ee.alph.ineq} follows from $|\alpha_{siN}|,\ |\alpha_{bi}| \le 1$.
\end{proof}
\end{prop}

Thus, the impact of subarray and block array factors $\alpha_{siN}$, $\alpha_{bi}$  on the overall IUI leakage factor $\alpha_{iN}$ is factorized, which simplifies the analysis considerably. In particular, using \eqref{ee.alph.ineq}, $|\alpha_{iN}|\to 0$ if either $|\alpha_{siN}|\to 0$ or $|\alpha_{bi}|\to 0$. This can be exploited to cancel grating lobe's effect on the FP. To this end, let us consider the asymptotic ("massive") regime where $N \to \infty$ while $N_b$ is fixed (constant), under distinct AoAs.
\begin{prop}
\label{prop.alpha.asymp}
If $N_b$ is fixed and $\theta_1\neq \theta_i$, then the following asymptotic relationship holds for the block-portioned NULA:
\bal
\label{eq12}
\lim_{N\to\infty} |\alpha_{iN}| = \left\{\begin{array}{ll}
|\alpha_{bi}(\phi_k)|, & \exists k\in I_k:\   \theta_i=\phi_k\\
0, & \operatorname{otherwise}.
\end{array}\right.
\eal
\begin{proof}
Using \eqref{ee.alph.ineq},
\bal
\label{eq.lim.alpha.in}
\lim_{N\to\infty}|\alpha_{iN}|= 0\ \mbox{if}\ \theta_i \neq \theta_1\ \&\ \theta_i \neq \phi_k \ \forall k \in I_k
\eal
since, from \eqref{eq6}, $\lim_{N\to\infty}|\alpha_{siN}|= 0$  in this case. On the other hand, if $\theta_i =\phi_k$ for some $k \in I_k$, then $|\alpha_{iN}|= |\alpha_{bi}(\phi_k)|$ since $|\alpha_{siN}|=1$ in this case. Note from \eqref{eq.alpha.def} that
\bal
\alpha_{bi}(\phi_k) &=\frac{1}{N_b} \sum_{n=0}^{N_b-1} e^{j2\pi nD(\sin(\phi_k)-\sin(\theta_1))} \notag\\
\label{eq13}
&=\frac{1}{N_b}\frac{\sin(\pi N_bk\Delta D/d)}{\sin(\pi k\Delta D/d)} e^{j\pi N_bk\Delta D/d}
\eal
where the last equality is from \eqref{eq.phi_k} and $D=(N-1)d + \Delta D$. Thus, $\alpha_{bi}(\phi_k)$ is independent of $N$ and this proves the 1st case in \eqref{eq12}.
\end{proof}
\end{prop}

From \eqref{eq12}, the FP is guaranteed under distinct AoAs if users do not align with grating lobes (or if GLs do not exist), $\theta_i \neq \phi_k\ \forall\ k\in I_k$. If some users do align, then the following equivalence holds:
\bal
\label{eq12.1}
\lim_{N\to\infty}\alpha_{iN}=0\quad \Leftrightarrow\quad \alpha_{bi}(\phi_k)=0\ \forall k\in I_k
\eal
i.e. grating lobes are canceled and the FP holds under \textit{any} distinct AoAs if $\alpha_{bi}(\phi_k)=0\ \forall k\in I_k$. The latter can be achieved by exploiting the NULA structure and choosing appropriate values of $N_b$ and $\Delta D$ as shown below.

To this end, we need the following concepts from number theory \cite[p. 231]{Rosen-99}:
\begin{itemize}
\item \emph{Greatest common divisor} of two integer $m$ and $n$, $\operatorname{gcd}(m,n)$: the largest positive integer that devides $m$ and $n$ without remainder; e.g. $\operatorname{gcd}(15,12)=3$.
\item \emph{Coprime (relative prime)}: two numbers $n$ and $m$ are coprime if $\operatorname{gcd}(n,m)=1$ (no common divisors); e.g. $\operatorname{gcd}(4,5)=1$, so, $4$ and $5$ are coprime. If $\operatorname{gcd}(n,m)=1$ and $n>1$, then $m/n$ is not integer.
\end{itemize}

The following theorem presents the NULA design to cancel all GLs and to achieve the FP under any distinct AoAs and beam steering subject to $|\theta_1| \le \theta_{max}$ for a given maximum steering angle $\theta_{max} \le \pi/2$, where $\theta_{max} = \pi/2$ corresponds to no constraint on steering.

\begin{thm}
\label{thm.NULA}
In LOS environment, the FP holds asymptotically for the NULA comprised of $N_b$ subarrays $\operatorname{ULA}(N,d)$, as in Fig. \ref{fig2}, with any fixed element spacing $d>0$ and any distinct users' AoAs, $\theta_1 \neq \theta_i$, $|\theta_1| \le \theta_{max}$, if:
\bal
\label{eq.T1}
\operatorname{(a)}\ N_b>\lfloor d(1+\sin\theta_{max}) \rfloor \ \operatorname{and\ (b)}\ \Delta D = pd/N_b,
\eal
where $p$ is a positive integer coprime with $N_b$, i.e. $\operatorname{gcd}(p,N_b)=1$, and there is no backward radiation.
\end{thm}
\begin{proof}
From \eqref{eq12}, the FP holds if $\alpha_{bi}(\phi_k)=0\ \forall k \in I_k$. To ensure this, we use \eqref{eq13} and find a proper $\Delta D$ so that
\bal
\operatorname{(i)}\ \sin(\pi N_bk\Delta D/d)=0\ \&\ \operatorname{(ii)}\ \sin(\pi k\Delta D/d)\neq 0
\eal
for all $k \in I_k$. (i) is equivalent to $N_b\Delta D/d$ being an integer:
\bal
\label{eq15}
N_b\Delta D/d = p\quad \Rightarrow\quad \Delta D = pd/N_b,\quad p=1,2,\cdots
\eal
However, one has to ensure (ii) as well with the following implication:
\bal
\label{eq17}
\sin(\pi k\Delta D/d)\neq 0\ \Leftrightarrow \frac{k\Delta D}{d} =\frac{pk}{N_b}\neq n' \quad  \forall k\in I_k
\eal
where $n'\in\{\pm1,\pm2,\cdots\}$ and the equality follows from \eqref{eq15}. To this end, we show that \eqref{eq.T1} imply $pk/N_b\neq n'\ \forall k\in I_k$. We need the following technical Lemmas.

\begin{lemma} \cite[p. 231, \mbox{fact }5]{Rosen-99}\\
\label{lem2}
If $a$ and $b$ are integers with $\operatorname{gcd}(a,b) = d$, then $\operatorname{gcd}(a/d,b/d) = 1$.
\end{lemma}

\begin{lemma}
\label{lem3}
If $a$ and $b$ are integers with $\operatorname{gcd}(a,b)=1$ and $c$ is a divisor of $a$, i.e. $\operatorname{gcd}(a,c)=c$, then $\operatorname{gcd}(a/c,b)=1$
\end{lemma}

\begin{lemma}\cite[p. 231, \mbox{fact }9]{Rosen-99}\\
\label{lem4}
If $a, b$, and $c$ are integers with $\operatorname{gcd}(a,b) = \operatorname{gcd}(a, c) = 1$, then $\operatorname{gcd}(a,bc) = 1$.
\end{lemma}

Now, assume that \eqref{eq.T1} holds and let $p_k$ be the greatest common divisor of $N_b$ and $|k|$:
\bal
\label{asmp2}\operatorname{gcd}(N_b,|k|)=p_k
\eal
Using \eqref{eq.T1}(a) and $|k|\le\ \lfloor d(1+|\sin\theta_1|) \rfloor$, which follows from \eqref{eq7} and \eqref{eq8}, one obtains
\bal
p_k \le |k| \le \lfloor d(1+|\sin\theta_1|) \rfloor <N_b
\eal
since $|\theta_1| \le \theta_{max} \le \pi/2$, which implies $N_b > |k|\ \forall k\in I_k$ . Using \eqref{asmp2} and Lemma 1, one obtains
\bal
\label{eq18}
\operatorname{gcd}(N_b,|k|)=p_k \quad\Rightarrow\quad \operatorname{gcd}(N_b/p_k,|k|/p_k)=1
\eal
Using $\operatorname{gcd}(N_b,p)=1$ and Lemma \ref{lem3},
\bal
\label{eq19}
 \operatorname{gcd}(N_b,p)=1\quad \Rightarrow\quad \operatorname{gcd}(N_b/p_k,p)=1
\eal
where $p_k$ is a divisor of $N_b$. Next, using \eqref{eq18}, \eqref{eq19}, and Lemma \ref{lem4}, one obtains:
\bal
\label{eq20}
\operatorname{gcd}(N_b/p_k, p|k|/p_k)=1
\eal
which means that $N_b/p_k$ and $p|k|/p_k$ are co-prime numbers and therefore their ratio is not an integer,
\bal
\label{eq21}
\frac{p|k|/p_k}{N_b/p_k}=\frac{p|k|}{N_b}\neq |n'|
\eal
This proves \eqref{eq17} and hence Theorem \ref{thm.NULA}. The no backward radiation condition is needed to eliminate the opposite AoAs $\theta_1$ and $\pi- \theta_1$ that cannot be discriminated by a NULA made of isotropic elements. It is always satisfied if $-\pi/2 \le \theta_1,\ \theta_i \le \pi/2$, as assumed here.
\end{proof}

\section{Discussion and Extensions}

Intuitively, the condition in \eqref{eq.T1}(a) on the number of subarrays $N_b$ is needed to make sure that there are enough nulls in the block array factor (i.e. the array factor of $\operatorname{ULA}(N_b,D)$, where each subarray is replaced by an isotropic element) to cancel all GLs. The condition \eqref{eq.T1}(b) on $\Delta D$ is needed to make sure that those nulls align with GLs and cancel them (if $p$ is not a coprime of $N_b$, then some nulls may not exist).

To illustrate Theorem 1, we apply it to the example in Fig. \ref{grating_lobe} and obtain $N_b>\lfloor 2d \rfloor=1$ so that setting  $N_b=25, p=21$ (coprime with $N_b$), $\Delta D = pd/N_b$ satisfies the conditions of Theorem 1 and, as Fig. \ref{grating_lobe} shows, partially cancels the only grating lobe even with finite $N_bN=100$ while preserving the main beam.

In general, larger $d$ calls for larger number $N_b$ of subarrays, but the same $N_b$ and $p$ can fit various $d$; under the latter condition, $\Delta D$ is proportional to $d$. Notably, while the actual directions and the number of GLs do depend on $\theta_1$, see \eqref{eq.phi_k}, \eqref{eq7}, \eqref{eq8}, the design of Theorem 1 is independent of $\theta_1$, i.e., it is not affected by beam steering and it cancels all GLs for any $\theta_1$. This provides the needed flexibility as it allows for beam steering without the need to change the array geometry. It can also be used to show (using the superposition principle) that all the above results, including Theorem 1, also hold in multipath and blocked-LOS channels, provided that different users have distinct AoAs (where each user is now allowed to have multiple AoAs to represent multiple paths). Since $N_b$ and $p$ in \eqref{eq.T1} are not unique, they can be further optimized to improve the performance for finite $N$.

The above results are straightforward to extend to directional (rather then omnidirectional) array elements, as typical in practice. For example, assume that the element pattern is that of a short dipole, $F_e(\theta)=\cos\theta$ (other weakly-directional elements, e.g. microstrip patches, have similar patterns \cite{Hansen-98}). It is straightforward to see that $\alpha_{iN}' = \alpha_{iN} \cos\theta_1 \cos\theta_i$, where $\alpha_{iN}'$ is the inter-user interference leakage term under directional elements and $\alpha_{iN}$ is that for isotropic elements. Clearly, $\alpha_{iN}' \to 0$ if and only if $\alpha_{iN} \to 0$ unless $\cos\theta_i=0$ (endfire directions only). Hence, while weakly-directional elements are able to somewhat reduce grating lobes, they are not able to eliminate them completely and thus have no impact on the FP (except for the endfire/backplane directions).


\end{document}